\documentclass[10pt,twocolumn,journal]{IEEEtran}%

 \newtheorem{theorem}{Theorem}

 \newenvironment{proof}{\begin{trivlist} \item[]{\em Proof.}}{\end{trivlist}}

\usepackage{amssymb, latexsym}
\usepackage{setspace}
\usepackage{lineno}

\usepackage{cite}

\ifCLASSINFOpdf
  \usepackage[pdftex]{graphicx}
  \graphicspath{{../pdf/}{../jpeg/}}
  \DeclareGraphicsExtensions{.pdf,.jpeg,.png}
\else
   \usepackage[dvips]{graphicx}
   \graphicspath{{../eps/}}
   \DeclareGraphicsExtensions{.eps}
\fi
\usepackage[cmex10]{amsmath}
\usepackage{algorithm}
\usepackage{algorithmic}

\usepackage{array}
\usepackage{fixltx2e}

\usepackage{stfloats}

\begin{document}

\singlespacing

%
\title{Signal periodic decomposition with conjugate subspaces}
%
%
%

\author{Shi-wen~Deng*,
        Ji-qing~Han*,~\IEEEmembership{Member,~IEEE,}
        \thanks{Manuscript received May X, 2016; revised XX XX, 2016.}

\IEEEcompsocitemizethanks{
\IEEEcompsocthanksitem *Deng~Shi-Wen is with School of Mathematical Sciences, Harbin Normal University, Harbin, China (e-mail: dengswen@gmail.com).\protect
\IEEEcompsocthanksitem *Han~Ji-Qing is with School of Computer Science and Technology, Harbin Institute of Technology, Harbin, China (e-mail: jqhan@hit.edu.cn).\protect\\
}
}

%
%

\markboth{IEEE transactions on signal processing,~Vol.~X, No.~X, XX~2016}%
{Shell \MakeLowercase{\textit{et al.}}: Bare Demo of IEEEtran.cls for Computer Society Journals}
%



\maketitle

\begin{abstract}
\boldmath

In this paper, we focus on hidden period identification and the periodic decomposition of signals.
Based on recent results on the Ramanujan subspace, we reveal the conjugate symmetry of the Ramanujan subspace with a set of complex exponential basis functions and represent the subspace as the union of a series of conjugate subspaces.
With these conjugate subspaces, the signal periodic model is introduced to characterize the periodic structure of a signal.
To achieve the decomposition of the proposed model, the conjugate subspace matching pursuit (CSMP) algorithm is proposed based on two different greedy strategies.
The CSMP is performed iteratively in two stages.
In the first stage, the dominant hidden period is chosen with the periodicity strategy.
Then, the dominant conjugate subspace is chosen with the energy strategy in the second stage.
Compared with the current state-of-the-art methods for hidden period identification, the main advantages provided by the CSMP are the following:
(i) the capability of identifying all the hidden periods in the range from $1$ to the maximum hidden period $Q$ of a signal of any length, without truncating the signal;
(ii) the ability to identify the time-varying hidden period with its shifted version; and
(iii) the low computational cost, without generating and using a large over-complete dictionary.
Moreover, we provide examples and applications to demonstrate the abilities of the proposed two-stage CSMP algorithm, which include hidden period identification, signal approximation, time-varying period detection, and pitch detection of speech.

\end{abstract}



%
\IEEEpeerreviewmaketitle

\section{Introduction}

Period estimation (detection) or periodicity identification is a fundamental problem in the field of signal processing.
The periodicity of a signal refers to the fact that it repeats after a certain duration of time.
Specifically, the signal $x[n]$ has the period $q$ that is the \textbf{smallest} positive integer satisfying
$x[n+q]=x[n], \text{ } \forall n \in \mathbb{Z}$.
In the more complex case, the periodicity of the signal cannot be observed directly, and it is the sum of several periodic signals with different periods, which are referred to as the hidden periods.
However, traditional methods such as the discrete Fourier transform (DFT), periodigram and autocorrelation cannot effectively identify hidden periods in signals \cite{Sethares99, Pei2015, Vaidyanathan2014_1, Vaidyanathan2014_2}.
We first formulate the problem of identifying hidden periods as follows.

\textit{ Problem 1:
A signal $x[n]$ of length $N$ is generated by a sum of $L$ signals with periods $q_1, \cdots, q_L$, where $N \ll lcm(q_1, \cdots, q_L)$.
How do we identify the hidden periods $q_1, \cdots, q_L$ in the signal $x[n]$?
More generally, given any signal $x[n]$ of length $N$ and the maximum hidden period $Q$, can it be approximated as a sum of periodic components with periods $q \in [1, \cdots, Q]$ and the approximation error?
}

Recently, a variety of approaches have been proposed to address the above problem.
Considering the limitations of the DFT for period estimation, Epps \textit{et al.} \cite{Epps2008, Epps2009} extended traditional DFT and autocorrelation and to propose the integer DFT and hybrid methods for identifying the hidden periods.
Unfortunately, the integer DFT cannot characterize all of the periodic behaviors of the signals.
Sethares and Staley \cite{Sethares99} proposed the periodicity transform to extract the periodic components of a signal by constructing 'periodic subspaces' into which the signal is projected.
Due to the ambiguity of the definition of the period, the 'periodic subspaces' in \cite{Sethares99} cannot exactly characterize the sets of periodic signals, and hence, their periodic decomposition of a signal cannot correctly identify the hidden periods and depended on the order of extraction of the periodic components.
To eliminate the drawbacks in \cite{Sethares99}, Muresan and Parks \cite{Muresan2003} proposed exactly periodic subspace decomposition (EPSD) by generating a series of orthogonal periodic subspaces based on the calculation of the intersections of the 'periodic subspaces' in \cite{Sethares99}.
However, EPSD cannot achieve orthogonal decomposition for all the periodic components due to the finite length of the signal.
In fact, only the periodic components whose periods are divisors of the length of the signal can be orthogonal decomposed.
A detailed discussion of the above algorithms can be found in our recent work \cite{RSP2015}.

More recently, Ramanujan sums were applied to analyze the periodic behaviors of signals.
Planat \textit{et al.} \cite{Planat2002,Planat2009} characterized the periodicity of the $1/f$ signal with the Ramanujan Fourier Transform (RFT).
Unfortunately, the RFT suffers the drawback of shift sensitivity \cite{Vaidyanathan2014_2}.
To eliminate this drawback, Vaidyanathan \cite{Vaidyanathan2014_2} proposed the Ramanujan Periodic Transform (RPT) based on the Ramanujan subspaces introduced by the same author in \cite{Vaidyanathan2014_1}.
By projecting the signal into a series of mutually orthogonal Ramanujan subspaces associated with the hidden periods of the signal, the RPT extracts the corresponding periodic components, which are also mutually orthogonal.
For a signal of finite length $N$, however, the RPT only generates Ramanujan subspaces whose periods are divisors of $N$, which means that RPT can only identify the hidden periods satisfying this condition.
Tenneti and Vaidyanathan \cite{Tenneti2015P,Vaidyanathan2015E} proposed the Ramanujan Filter Bank (RFB) based on the Ramanujan subspaces to identify all the hidden periods of a signal.
Although the RFB can identify all the hidden periods, even those that change with time, the identification results of the RFB contain some false hidden periods due to the overlaps of the filter banks.
Moreover, Pei and Lu \cite{Pei2015} introduced the intrinsic integer-periodic function for identifying hidden periods based on the Ramanujan subspaces.

Some other approaches were proposed to identify the hidden periods based on the representation of a signal over the redundant dictionary.
Nakashizuka \textit{et al.} \cite{Nakashizuka2008} and Vaidyanathan \textit{et al.} \cite{Vaidyanathan2014C} proposed methods for the periodic decomposition of signals based on the framework of the spare representation of a signal.
In \cite{Vaidyanathan2014C}, the Farey dictionary is generated based on the Farey sequence, over which the sparse representation of the signal is obtained.
It was further extended to the Ramanujan dictionary-based approaches \cite{Tenneti2015-N,Vaidyanathan2015-M}.
However, one of the serious drawbacks of these approaches is that the dictionary will become too large when the expected maximum hidden period $Q$ is large, which means that a very large computational cost is required for the decomposition over the dictionary.
The dimension of the dictionary must be greater than $\sum_{1=1}^Q \phi(q)$, where $\phi(q)$ is the \textit{Euler's totient function} of $q$ and is also the dimension of the Ramanujan subspace associated with the period $q$.
For example, when $Q=512$, the number of columns of the dictionary is $79852$!.

In this paper, we present the method of signal periodic decomposition over conjugate subspaces based on the greedy strategy, named conjugate subspace matching pursuit (CSMP), which can efficiently and effectively solve \textit{Problem 1}.
The CSMP is a subspace pursuit method that was first introduced in \cite{Goodwin1999} to obtain a better representation of a signal in the time-frequency plane.
However, the CSMP proposed in this paper is completely different from the traditional matching pursuit algorithms \cite{Goodwin1999, Deng2001} used in signal decomposition in terms of constructing the dictionary and the greedy strategy, which are the key problems for the matching pursuit algorithm.
To identify the hidden periods and to perform the periodic decomposition of signals, the proposed method is based on the results for the Ramanujan subspace in \cite{Vaidyanathan2014_1} and our recent work \cite{RSP2015}.
First, we generate the Ramanujan subspace with the complex exponential basis from the frequency point of view and reveal that the Ramanujan subspace has conjugate symmetry structure.
Second, based on the symmetry structure, the Ramanujan subspace is represented as the union of a series of conjugate subspaces.
The union of the conjugate subspaces associated with all the hidden periods can be used as the dictionary for the periodic decomposition.
Third, unlike the dictionary-based method \cite{Vaidyanathan2014C,Tenneti2015-N} or a traditional matching pursuit algorithm such as \cite{Goodwin1999, Deng2001}, we perform a two-stage CSMP without constructing and using the whole dictionary.
In the first stage, the most dominant period of the signal in the current iteration is chosen based on the periodicity metric defined in \cite{RSP2015} without calculating the projection of the signal onto each Ramanujan subspace.
In the second stage,
the projections of the signal into the conjugate subspaces belonging to the Ramanujan subspace associated with the chosen period are calculated, and the component with the largest projection energy is removed from the current signal.
When the CSMP is stopped, the sums of the projections belonging to the same Ramanujan subspace, denoted $\mathcal{S}_q$, are the periodic components with the hidden period $q$.

Compared with the current state-of-the-art methods for hidden period identification, the main advantages provided by the CSMP are the following:
(i) the capability of identifying all the hidden periods in the range from $1$ to the maximum period $Q$ of a signal of any length, without truncating the signal or the dictionary;
(ii) the ability to identify the time-varying hidden period with its shifted version; and
(iii) the low computational cost, without generating and using the large over-complete dictionary.

%
%
%
%
%
%
%

The paper is organized as follows.
We begin with a brief review of some necessary concepts and results of the Ramanujan subspace in Section \ref{SEC:RS}.
In Section \ref{SEC:CS}, we redefine the Ramanujan subspace with the complex exponential basis, reveal its complex conjugate symmetry, and provide the representation of the Ramanujan subspace with its conjugate subspaces.
The general model for the periodic decomposition based on the conjugate subspaces is introduced in Section \ref{SEC:SPCS}.
Section \ref{SEC:ALG} presents the CSMP algorithm to perform the signal periodic decomposition.
With the proposed method, some examples and real applications are provided in Section \ref{SEC:EA}.
We provide conclusions in Section \ref{SEC:CON}.

\section{Ramanujan subspace}
\label{SEC:RS}

In this section, we briefly review some necessary concepts and results of the Ramanujan subspace, which first appeared in \cite{Vaidyanathan2014_1}.
The Ramanujan subspace is constructed based on the Ramanujan sums.
For any positive integer $q$, the Ramanujan sums is a sequence with period $q$ and is defined as follows
\begin{align*}
	c_q(n)=\underset{(k,q)=1}{\sum_{k=1}^q} e^{j2 \pi kn/q}, \text{ for } n= \cdots, -1, 0, 1, \cdots
\end{align*}
where $(k,q)$ denotes the greatest common divisor (gcd) of $k$ and $q$ and $(k,q)=1$ means that $k$ and $q$ are coprime.
Then, the Ramanujan subspace $\mathcal{S}_q$ is defined by the column space of the following integer circulant matrix $\mathbf{B}_q$
\begin{align*}
        \small
	\mathbf{B}_q = \left[
		\begin{array}{cccc}
			c_q(0) & c_q(q-1)&  \cdots& c_q(1) \\
			c_q(1) & c_q(0)&  \cdots& c_q(2) \\
			c_q(2) & c_q(1)&  \cdots& c_q(3) \\
            \vdots & \vdots &  \ddots & \vdots \\
			c_q(q-2) & c_q(q-3)&  \cdots& c_q(q-1) \\
			c_q(q-1) & c_q(q-2)&  \cdots& c_q(0)
		\end{array}
	\right]
\end{align*}
where $c_q(\cdot)$ is the Ramanujan sums.

According to \cite{Vaidyanathan2014_1,Vaidyanathan2014_2,Pei2015},
the Ramanujan subspace $\mathcal{S}_q$ is capable of characterizing the periodic component of period $q$ of the signal.
Moreover, the relationship between the Ramanujan subspace and the DFT matrix proposed in \cite{Vaidyanathan2014_1} is summarized as follows.
\begin{theorem} \label{THEOREM:RS}
	The Ramanujan subspace $\mathcal{S}_q \subset \mathbb{C}^q$ is identical to the space spanned by those $\phi(q)$ columns of the $q \times q$ DFT matrix, whose column indices $k$ are coprime with $q$.
\end{theorem}

Note that when $N$ is the integer multiple of $q$, say $N=qM$, the Ramanujan subspace $\mathcal{S}_{qM} \subset \mathbb{C}^N$ is also denoted as $\mathcal{S}_q$.
Therefore, the Ramanujan subspace $\mathcal{S}_q$ can characterize the periodic behavior of the signal whose length is the integer multiple of the period $q$.
The definition of the Ramanujan subspace will be extended in the following section.


\section{Conjugate subspaces of the Ramanujan subspace}
\label{SEC:CS}

In this section, we first redefine the Ramanujan subspace with a set of complex exponential basis functions and reveal its complex conjugate symmetry structure.
Then, the conjugate subspace of the Ramanujan subspace are defined by a pair complex exponential basis functions based on their complex conjugate symmetry.
Next, the Ramanujan subspace is represented with a series of conjugate subspaces.

\subsection{Complex conjugate symmetry of the Ramanujan subspace}

The Ramanujan subspace $\mathcal{S}_q$ of the period $q$ is a linear subspace with dimension $\phi(q)$.
According to Theorem \ref{THEOREM:RS}, from the frequency point of view, it consists of $\phi(q)$ center frequencies, which are
\begin{align} \label{Eq:FreqOmega}
    \omega_{q,i}=2\pi\frac{k_i}{q}, \text{  } i=1,\cdots,\phi(q)
\end{align}
where the positive integer $k_i$ is coprime to $q$ and satisfies $1 \leq k_i<k_{i+1} < q$.

Instead of using the integer basis based on the Ramanujan sums $c_q[n]$, the Ramanujan subspace $\mathcal{S}_q$ can be defined based on the complex exponential functions
\begin{align} \label{Eq:CFBasis}
    g(\omega_{q,i})=R_{N,q,i}e^{jn\omega_{q,i}}, \text{  } n=0,\cdots,N-1
\end{align}
where $R_{N,q,i}$ is a constant associated with the frequency $\omega_{q,i}$ to obtain the unit-norm function satisfying $\|g(\omega_{q,i})\|=1$, $n$ is the time (or sample) index, and $N$ is the signal length.
Thus, the Ramanujan subspace $\mathcal{S}_q$ can be redefined by the set of basis functions
$\{ g(\omega_{q,i})\}_{i=1}^{\phi(q)}$, that is,
\begin{align} \label{EQ:SqCE}
    \mathcal{S}_q \triangleq \text{span} \left \{ g(\omega_{q,i}) \right\}, \text{  } i=1,\cdots,\phi(q)
\end{align}
Instead of restricting the signal length, $N$ must be an integer multiple of the period $q$, and the Ramanujan subspace in (\ref{EQ:SqCE}) can be applied to signals of any length, which is an extension of the traditional definition of the Ramanujan subspace.
It is worthwhile mentioning that this extension is achieved at the expense of losing the orthogonality of the basis $\{ g(\omega_{q,i})\}_{i=1}^{\phi(q)}$.

As the basis $\{g(\omega_{q,i})\}_{i=1}^{\phi(q)}$ corresponds to the frequencies $\{\omega_{q,i}\}_{i=1}^{\phi(q)}$, the Ramanujan subspace $\mathcal{S}_q$ has a complex conjugate symmetry structure.
For example, when $q=9$, the frequencies contained in $\mathcal{S}_q$ are $\{ 2\pi \frac{1}{9},2\pi\frac{2}{9},2\pi\frac{4}{9}, 2\pi\frac{5}{9},2\pi\frac{7}{9},2\pi\frac{8}{9}\}$.
Fig. \ref{FIG:CP} shows these frequencies distributed in an unit circle with a maximum frequency of $2 \pi$.
Obviously, these frequencies
$2\pi\frac{1}{9}$ and $2\pi\frac{8}{9}$,
$2\pi\frac{2}{9}$ and $2\pi\frac{7}{9}$,
and $2\pi\frac{4}{9}$ and $2\pi\frac{5}{9}$
are symmetric about the horizontal axis.
The symmetry among these frequencies implies that the corresponding complex exponential functions $\{g(\omega_{q,i})\}_{i=1}^{\phi(q)}$ also have complex conjugate symmetry structure.
The following theorem indicates that the Ramanujan subspace $\mathcal{S}_q$ defined by the complex exponential basis has complex conjugate symmetry structure.

\begin{figure}[h]
        \begin{minipage}[h]{1.0\linewidth}
          \centering
          \centerline{\includegraphics[width=7.5cm]{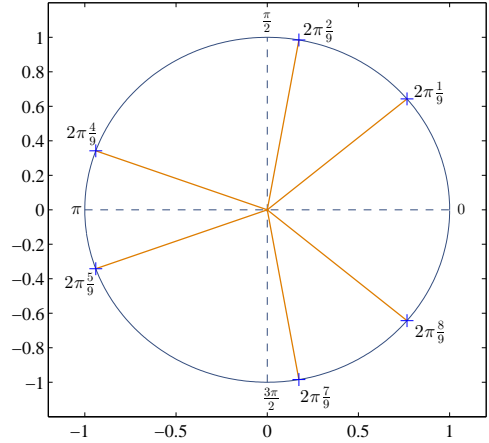}}
        \end{minipage}
        \caption{The frequencies are contained in the Ramanujan subspace $\mathcal{S}_9$.}
        \label{FIG:CP}
\end{figure}

\begin{theorem} \label{THEOREM:CONJ}
    Given the period $q$, the Ramanujan subspace $\mathcal{S}_q$ represented by the complex exponential basis $\{ g(\omega_{q, i})\}_{i=1}^{\phi(q)}$ has complex conjugate symmetry, satisfying the following:
    \begin{align}
        \overline{g_q(\omega_{q,i})}&=g_q(\omega_{q,\phi(q)-i}), \text{ if } q \geq 3 \\
        \overline{g_q(\omega_{q,i})}&=g_q(\omega_{q,i}),  \text{ if }  1 \leq q \leq 2
    \end{align}
    where the frequency $\omega_{q,\phi(q)-i}=2\pi\frac{q-k_i}{q}$.
    The complex exponential function $g_q(\omega_{q,i})$ and its complex conjugate $\overline{g_q(\omega_{q,i})}$ are referred to as a complex conjugate pair.
\end{theorem}
\begin{proof}

    If $q \leq 2$, then $\phi(q)=1$.
    For this case, there exists only one frequency component  in $\mathcal{S}_q$, and $g_q(\omega_{q,1})$ is a real function (or vector); hence, $\overline{g_q(\omega_{q,1})}=g_q(\omega_{q,1})$.
    Specifically, $\omega_{1,1}=2\pi$ and $\omega_{2,1}=\pi$, for $q=1$ and $q=2$.

    If $q>2$, then $\phi(q)$ is even and $\phi(q)\geq 2$.
    Note that if $(k_i,q)=1$, it follows that $(q-k_i,q)=1$ as well.
    For the frequency component $\omega_{q, i}$ contained in $\mathcal{S}_q$, we have
    \begin{align*}
	   \overline{e^{j n \omega_{q,i}}} &= e^{-j n \omega_{q,i}} \\
                                    &= cos( n \omega_{q,i}) - j sin( n \omega_{q,i}) \\
                                    &= cos(2\pi n- n \omega_{q,i}) + j sin(2\pi n- n \omega_{q,i}) \\
                                    &= e^{j2 \pi n \frac{q-k_i}{q}} \\
                                    &= e^{j 2\pi n \omega_{q,{\phi(q)-i}}}
    \end{align*}
    where $k_{\phi(q)-i}=q-k_i$.
   Hence,
    \begin{align*}
        \overline{g_q(\omega_{q,i})}=g_q(\omega_{q,{\phi(q)-i}})
    \end{align*}
    is proven.
\end{proof}

Theorem \ref{THEOREM:CONJ} reveals an important property of the Ramanujan subspace $\mathcal{S}_q$ when it is defined by a set of complex exponential basis.
This means that the Ramanujan subspace $\mathcal{S}_q$ can be represented by a series of complex conjugate subspaces.

\subsection{Representation of the Ramanujan subspace with conjugate subspaces}

With Theorem \ref{THEOREM:CONJ}, we know that the complex exponential basis of the Ramanujan subspace $\mathcal{S}_q$ contains a series of complex conjugate pairs.
Let $g(\omega_{q,i})$ and $\overline{g(\omega_{q,i})}$ be a complex conjugate pair.
The subspace $\mathcal{G}_{q,i}$ is referred to as the complex conjugate subspace (CCS) of $\mathcal{S}_q$, defined by
\begin{align} \label{EQ:G}
	\mathcal{G}_{q,i}=\text{span} \left\{ g(\omega_{q,i}), \overline{g(\omega_{q,i})} \right\},
\end{align}
which is completely determined by $g(\omega_{q,i})$ and its  complex conjugate $\overline{g(\omega_{q,i})}$.

The complex conjugate pair is not mutually orthogonal when the length $N$ is not an integer multiple of the period $q$.
Let $|c(q,i)| \in [0,1)$ denote the absolute correlation coefficient between $g(\omega_{q,i})$ and $\overline{g(\omega_{q,i})}$, where $c(q,i) \triangleq \langle g(\omega_{q,i}),\overline{g(\omega_{q,i})} \rangle$.
Fig. \ref{FIG:CCCP} shows the correlation coefficients of these complex pairs of the basis of the Ramanujan subspace $\mathcal{S}_9$ as functions of the length $N$.
The correlation coefficient satisfies $|c(q,i)|=0$ only when the length $N$ is an integer multiple of the period $9$.
All of the values of these correlation coefficients decrease when the length $N$ increases, which means that the orthogonality of the conjugate pair increases.

\begin{figure}[t]
        \begin{minipage}[h]{1.0\linewidth}
          \centering
          \centerline{\includegraphics[width=9.5cm]{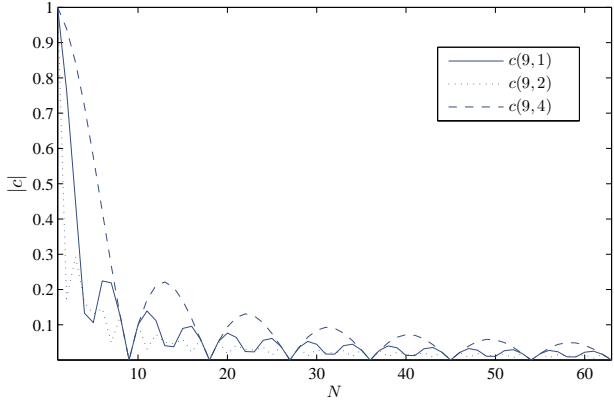}}
        \end{minipage}
        \caption{The correlation coefficients of the conjugate pairs of $\mathcal{S}_9$ as the functions of length $N$.}
        \label{FIG:CCCP}
\end{figure}

According to Theorem \ref{THEOREM:CONJ},
there are $M_q$ conjugate pairs in the complex exponential basis of $\mathcal{S}_q$, where
\begin{align}\label{EQ:Mq}
	M_q = \left\{
		\begin{array}{rl}
			1, &\text{  if } q \leq 2 \\
			\phi(q)/2, &\text{  if } q \geq 3
		\end{array}
	\right.
\end{align}
Thus, the Ramanujan subspace $\mathcal{S}_q$ contains $M_q$ CCSs corresponding to these complex conjugate pairs and can be represented as the union of these CCSs as follows
\begin{align} \label{EQ:RSandCS}
	\mathcal{S}_q = \bigcup_{i=1}^{M_q}  \mathcal{G}_{q,i}
\end{align}
where the symbol ``$\bigcup$'' denotes the union of these CCSs.
Specially, the Ramanujan subspaces $\mathcal{S}_q \subseteq \mathcal{R}^N$ can be represented as the direct sum ``$\bigoplus$'' of these CCSs,
as they are mutually orthogonal when $N$ is an integer multiple of $q$.
In general, when $N$ is not an integer multiple of $q$, these CCSs are not mutually orthogonal, and the intersection of the different CCSs contains nonzero elements.

By representing the Ramanujan subspace $\mathcal{S}_q$ with the CCSs in (\ref{EQ:RSandCS}), the periodic component $x_q \in \mathcal{S}_q$ can be represented with these CCSs.
We will construct this representation in the following section.

\section{Signal periodic model with conjugate subspaces}
\label{SEC:SPCS}

In this section, we first introduce the signal periodic model for formulating the aforementioned \textit{Problem 1} and reformulate the model with CCSs.
Then, the optimal problem for solving the signal periodic model is presented.
With the proposed model, the periodic structure of a signal is explored.

We formulate the aforementioned \textit{Problem 1} as the following signal periodic model.
Given the maximum hidden period $Q$, we assume that the signal $x \in \mathbb{R}^N$ can be represented as
\begin{align}\label{EQ:PeriodRe}
	x = \sum_{q=1}^Q x_q + r
\end{align}
where $x_q$ is the periodic component with the hidden period $q$, which is the projection of the signal $x$ onto the Ramanujan subspace $\mathcal{S}_q$, and $r$ is the residual error.
We refer to the representation in (\ref{EQ:PeriodRe}) as the signal periodic model or signal periodic decomposition.
With the signal periodic model in (\ref{EQ:PeriodRe}), the strengths of the hidden periodic components of the signal $x$ can be easily observed by their energies $\|x_q\|^2$, for $q=1,\cdots,Q$.
Similar to the energy spectrum defined in DFT, we refer to the strengths of these hidden periods as the periodic spectrum of the signal $x$, demonstrated using an example in Fig. \ref{FIG:IdenHidPs} (b).
Next, we will further represent the signal periodic model with the CCSs of the Ramanujan subspace corresponding to the hidden periods $1,\cdots,Q$.
Before doing so, we present the orthogonal projection of a signal into the CCS below.

Let $\mathcal{G}$ be the CCS spanned by the complex conjugate pair $\{g,\bar{g}\}$, that is, $\mathcal{G}=\text{span}\{g,\bar{g}\}$.
Because a real-world signal is real valued, we only consider the orthogonal projection of a real-valued signal in the CCS.
The orthogonal projection of the signal $x \in \mathbb{R}^N$ onto $\mathcal{G}$ can be represented by
\begin{align} \label{EQ:PCinCS}
    x_{\mathcal{G}} = 2 \text{Re}(\alpha g)
\end{align}
where $\text{Re}(z)$ represents the real part of the complex vector $z \in \mathbb{C}^N$.
The complex coefficient $\alpha \in \mathbb{C}$ is defined by
\begin{align} \label{EQ:ALPHA}
    \alpha \triangleq& \langle x, g \rangle_{*} \\
           =& \frac{1}{1-|c|^2}(\langle g,x \rangle - c \langle \bar{g},x \rangle)
\end{align}
where $c = \langle g,\bar{g} \rangle$ and $|c| \in [0,1)$ is the correlation coefficient between $g$ and its complex conjugate $\bar{g}$.
Here, we introduce the symbol $\langle x,g \rangle_*$ to denote the complex coefficient of the projection of signal $x$ onto the CCS $\mathcal{G}$.
Importantly, the projection $x_{\mathcal{G}}$ of $x$ onto the CCS $\mathcal{G}$ in (\ref{EQ:PCinCS}) can be represented by a single complex coefficient $\alpha$, which satisfies the following:
\begin{align}
    \|x_{\mathcal{G}}\|^2 = 2|\langle x,g \rangle_*| = 2|\alpha|^2
\end{align}
More details of the projection into the CCS are given in Appendix \ref{SEG:APP-A}.

The signal periodic model based on the CCSs is considered here.
With (\ref{EQ:RSandCS}) and (\ref{EQ:PCinCS}), the periodic component $x_q \in \mathcal{S}_q$ can be represented by the orthogonal projections of $x_q$ onto the CCSs $\{\mathcal{G}_{q,i}\}_{i=1}^{M_q}$ of $\mathcal{S}_q$ as follows
\begin{align} \label{EQ:PCCS}
	x_q=2\sum_{i=1}^{M_q} \text{Re}( \alpha_{q,i} g(\omega_{q,i}) )
\end{align}
where $\alpha_{q,i}$ is the orthogonal projection coefficient of $x_q$ onto the conjugate subspace $\mathcal{G}_{q,i}$.
(\ref{EQ:PCCS}) provides a way of representing the periodic component of a signal based on the CCSs of the corresponding Ramanujan subspace.
By substituting (\ref{EQ:PCCS}) into (\ref{EQ:PeriodRe}), the signal periodic model in (\ref{EQ:PeriodRe}) can be represented as
\begin{align} \label{EQ:SigCS}
	x = 2\sum_{q=1}^Q	\sum_{i=1}^{M_q}  \text{Re}\left(\alpha_{q,i}g(\omega_{q,i})\right) + r
\end{align}
 based on all the CCSs of the Ramanujan subspaces $\{ \mathcal{S}_q \}_{q=1}^Q$.

The signal periodic model in (\ref{EQ:SigCS}) reveals the intrinsic periodic structure of the signal through a linear combination of its hidden periodic components, but it cannot be directly achieved in a linear manner.
Let the complex matrix $\mathbf{D} \in \mathbb{C}^{N \times M} $ denotes the union of the complex exponential functions $\{ g(\omega_{q,i}) \}$ in (\ref{EQ:SigCS}), where $M$ is the total number of the CCSs and
\begin{align}
M =\sum_{q=1}^Q {M_q} = 2+\sum_{q=3}^Q\phi(q)/2
\end{align}
The complex vector $\pmb{\alpha} \in \mathbb{C}^M$ is stacked from the complex coefficients $\{\alpha_{q,i}\}$ using the rule
\begin{align}\label{EQ:MAPIdX}
    m=\sum_{p=1}^{q-1}M_p + i
\end{align}
where $m \in [1,M]$, $q \in [1, Q]$, and $i \in [1, M_q]$.
Thus, the signal periodic model in (\ref{EQ:SigCS}) can be written in matrix form as
\begin{align} \label{EQ:MDMATR}
	x = 2 \text{Re}(\mathbf{D}^T \pmb{\alpha}) + r
\end{align}
where the superscript $T$ denotes the matrix transposition (not the conjugate transposition).
Usually, for a finite-duration signal $x$ of length $N$, we have $M \gg N$.
This implies that $\mathbf{D}$ is very redundant and that there exist infinite solutions of the signal periodic model.
Therefore, we expect to find a sparse solution to the model in (\ref{EQ:MDMATR}) by solving the following optimal problem
\begin{align} \label{EQ:MPProb}
    \begin{array}{lr}
        \underset{\pmb{\alpha \in} \mathbb{C}^M}{\text{min }} \|\pmb{\alpha}\|_0 \\
        \text{ subject to } \|x - 2 \text{Re}(\mathbf{D}^T \pmb{\alpha})\|_2 \leq \epsilon
    \end{array}
\end{align}
where $\epsilon > 0$ is the error tolerance for the linear approximation.

We will present the CSMP algorithm to achieve the decomposition of the signal periodic model in the next section.
It is worthwhile mentioning that the over-complete dictionary $\mathbf{D}$ is never constructed and used in our algorithm.

\section{Conjugate subspace matching pursuit algorithm}
\label{SEC:ALG}

In this section, we present a two-stage CSMP algorithm for the decomposition of the signal periodic model.
The CSMP is iteratively performed using some greedy strategies.
Instead of using the common greedy strategies, the periodicity strategy and the energy strategy are proposed and used in the two stages of the CSMP in Subsection \ref{SEC:GS}.
In the first stage, the dominant hidden period is chosen using a periodicity strategy.
The dominant CCS is chosen using an energy strategy in the second stage.
We provide the details of the two-stage CSMP algorithm in Subsection \ref{SEC:TWCSMP}.

\subsection{Greedy strategies}
\label{SEC:GS}

In each iteration, the CSMP is to choose a suitable CCS based on a certain greedy strategy.
However, the commonly used strategies in MP \cite{Mallat1993,Deng2001,Goodwin1999} or the subspace pursuit \cite{Dai2009} algorithms are unsuitable for our CSMP algorithm.
With these strategies, the CSMP needs to calculate the projections of the current residual signal in all the subspaces spanned by each atom of the dictionary.
This means that a high computational cost is required to solve the optimal problem in (\ref{EQ:MPProb}), as the size of the dictionary $\mathbf{D}$ may be very large when the maximum hidden period $Q$ is large.
For example, when $Q=512$, the number of columns of $\mathbf{D}$ is $39927$, which is almost half the size of the Farey dictionary \cite{Vaidyanathan2014C} or nested periodic matrices used in \cite{Tenneti2015-N}.
On the other hand, the CCS chosen in each iteration for the CSMP needs to capture the dominant periodicity of the current residual signal rather than just match the component in the residual signal with the dictionary atom well.

We propose to perform the CSMP in two stages with different strategies, the periodicity strategy and energy strategy.
With the periodicity strategy in the first stage, the dominant hidden period is chosen according to a certain periodicity metric.
With the energy strategy in the second stage, the CCS is subsequently selected from the Ramanujan subspace corresponding to the chosen hidden period, where the current signal can obtain the maximum projection energy.

For the periodicity strategy in the first stage, the periodicity metric first introduced in our recent work \cite{RSP2015} is used to choose the dominant hidden period, which is defined as
\begin{align} \label{EQ:PMxq}
  	P({x}_q,q) \triangleq \frac{N+q}{2q} \|{x}_q\|^2
\end{align}
where $N$ is the signal length, $q$ is the measured hidden period, and $x_q$ is the hidden periodic component that is the projection of the signal $x$ onto the Ramanujan subspace $\mathcal{S}_q$.
Instead of directly calculate the projection energy $\|x_q\|^2$ of the hidden periodic component by projecting $x$ into the Ramanujan subspace $\mathcal{S}_q$, \cite{RSP2015} shows that $\|x_q\|^2$ can be iteratively calculated using
\begin{align}\label{EQ:ENEPSD}
	\| {x}_q \|^2 = \|\hat{x}_q\|^2-\sum_{p \in \Gamma} \|{x}_{p}\|^2
\end{align}
where $\Gamma$ is the set of all the divisors of $q$ expect for $q$ itself and $\|\hat{x}_q\|^2$ is the estimate of the periodic energy.
According to \cite{Wise1976,Muresan2003, RSP2015}, $\|\hat{x}_q\|^2$ can be estimated by
\begin{align}\label{EQ:EMLE}
	\|\hat{{x}}_q\|^2=\frac{q}{N}\left( \varphi_{{x}}(0)+2\sum_{l=1}^{M-1} \varphi_{{x}}(lq) \right)
\end{align}
where $\varphi_{x}(\cdot)$ is the autocorrelation function of $x$, $M=\lfloor N/q\rfloor$, and $\|\hat{x}_1\|^2=\|x_1\|^2$.
Thus, the dominant hidden period $q^*$ can be chosen by
\begin{align}
	q^* = \underset{q \in [1, Q]}{\text{ argmin }} P(x_q,q)
\end{align}
where $Q$ is the maximum period.
Importantly, it is unnecessary to calculate the periodicity metric in (\ref{EQ:PMxq}) for the hidden periods by projecting the signal onto each Ramanujan subspace $\{\mathcal{S}_1\}_{q=1}^Q$.
This greatly reduces the computational cost of the CSMP.

With the chosen hidden period $q^*$ in the first stage, the dominant CCS is chosen from all the CCSs
$\{ \mathcal{G}_{q^*,i} \}_{i=1}^{M_{q^*}}$ of the Ramanujan subspace $\mathcal{S}_{q^*}$,
where
\begin{align*}
    \mathcal{G}_{q^*,i}=span\{g(\omega_{q^*,i}),\overline{g(\omega_{q^*,i}})\}.
\end{align*}
Because the CCS $\mathcal{G}_{q^*,i}$ can be completely determined by only one basis function, the selection of the CCS is equivalent to choosing the basis function $g(\omega_{q^*,i})$.
According to the energy strategy, the dominant basis function is chosen for which
\begin{align}
	|\alpha_{q^*,i^*}| \geq |\alpha_{q^*,i}|, \text{ for } i=1, \cdots, M_{q^*}
\end{align}
where $\alpha_{q^*,i}$ is the complex projection coefficient of the signal $x$ onto the CCS $\mathcal{G}_{q^*,i}$.

By projecting the signal $x$ onto the CCS $\mathcal{G}_{q^*,i^*}$, the greedy strategies of both periodicity and energy are achieved simultaneously.
These strategies will be used in the following two-stage CSMP at each iteration.

\subsection{Two-stage CSMP algorithm}
\label{SEC:TWCSMP}

With the greedy strategies, the two-stage CSMP algorithm is carried out as follows.

Let the initial residual signal $r_0=x$ and $Q$ be the maximum hidden period of $x$.
Let $r_{l-1}$ denote the residual signal after $l-1$ iterations, which has already been computed in previous iterations.
In the $l$-th iteration,
the dominant hidden period $q_l$ of $r_{l-1}$ is chosen by using the periodicity strategy at the first stage.
Subsequently, the dominant CCS $\mathcal{G}_{q_l,i_l}$, characterized by the complex basis function $g(\omega_{q_l,i_l})$ of the Ramanujan subspace $\mathcal{S}_{q_l}$, is chosen with the energy strategy in the second stage.

Then, the projection of $r_{l-1}$ onto the dominant CCS $\mathcal{G}_{q_l,i_l}$, which is characterized by $g(\omega_{q_l,i_l})$, is removed to obtain the new residual signal $r_l$ in the $l$-th iteration, that is
\begin{align}
	{r}_{l}={r}_{l-1}-2 \text{Re} \left(\alpha_{q_l,i_l} g(\omega_{q_l,i_l}) \right)
\end{align}
where $\alpha_{q_l,i_l} = \langle r_{l-1}, g(\omega_{q_l,i_l})  \rangle_*$, satisfying
\begin{align}
\|{r}_{l}\|^2 = \|{r}_{l-1}\|^2 - 2 \|\alpha_{q_l,i_l}\|^2
\end{align}
For concise representation, let $\alpha_l = \alpha_{q_l,i_l}$ and $g_l=g(\omega_{q_l,i_l})$.
The signal $x$ can be represented with $L$ periodic components and the residual $r_{L+1}$ as follows
\begin{align}
	x = 2 \sum_{l=1}^L \text{Re}(\alpha_l g_l) + r_{L+1}
\end{align}
satisfying
\begin{align}
	\|x\|^2 = 2\sum_{l=1}^L \|\alpha_l\|^2 + \|r_{L+1}\|^2
\end{align}
The details of the algorithm are described in Algorithm \ref{Alg_CSMP}.

\begin{algorithm}[t]
    \caption{\small Conjugate subspace matching pursuit (CSMP)}

    \textbf{Input:} signal ${x}\in \mathbb{R}^N$ \\
    \textbf{Output:} complex projection coefficients $\alpha_1, \cdots, \alpha_L$ and complex basis functions $g_1, \cdots, g_L$

    \textbf{Initialize:} Set ${r}_0={x}$ \\
    \textbf{for} $l=1 \text{ to } L$ \textbf{do}

      \ \ \ \ \textit{Stage 1}:

      \ \ \ \ \ (1) Find the dominant hidden period of $r_{l-1}$ with the periodicity strategy
      \begin{align*}
        q_l = \underset{q \in [1,Q]}{\text{argmax}} P(x_q,q)
      \end{align*}

      \ \ \ \ \textit{Stage 2}:

      \ \ \ \ (2) Generate complex exponential functions
      $\mathbf{D}_{l}=\{ g(\omega_{q_l, i} \}_{i=1}^{M_{q_l}}$ according to $q_l$

      \ \ \ \ (3) Find the dominant CCS and projection coefficient with the energy strategy
      \begin{align*}
            g_l =& \underset{g \in \mathbf{D}_l}{\text{argmax}}
               \left\{ |\langle r_{l-1},g \rangle_*| \right\} \\
           \alpha_l =& \langle r_{l-1},g_l \rangle_*
      \end{align*}

      \ \ \ \ (4) Update residual ${r}_{l}={r}_{l-1}-2 \text{Re}(\alpha_l g_l)$

    \textbf{end} 

    \label{Alg_CSMP}

\end{algorithm}

\section{Examples and applications}
\label{SEC:EA}

In this section, we provide several examples and applications to demonstrate the abilities of our proposed two-stage CSMP algorithm.
They include hidden period identification, signal approximation, time-varying period detection, and pitch detection of speech.

\subsection{Identifying hidden periods}

Some signals are generated by the superposition of several periodic signals and hence contain hidden periods.
More generally and formally, we assume that the signal $x[n]$ of length $N$ is generated by a sum of $L$ signals with periods $q_1, \cdots, q_L$, where $N \ll lcm(q_1,\cdots, q_L)$, as described in \textit{Problem 1}.
The aim of the examples here is to illustrate the capability of the CSMP algorithm for identifying hidden periods, compared with three methods: RFT \cite{Planat2002,Planat2009}, RPT \cite{Vaidyanathan2014_2}, and EPSD \cite{Muresan2003}.

The synthetic signal $x[n]$ is composed of eight sinusoidal components and is defined by
\begin{align} \label{EQ:SYSIG}
	x[n] =\sum_{q \in \Gamma} cos \left(\frac{2\pi n}{q} \right)
\end{align}
where $\Gamma=\{ 5,12,25,26,57,58,70,85 \}$ is the set of hidden periods and $n$ is the sample index.
The CSMP and three other methods are used to identify the hidden periods of the signals of different lengths $650$ and $1950$.
Because the signal lengths $650$ and $1950$ are less than the least common multiple $127859550$ of these periods in $\Gamma$, it is difficult to directly observe the periodicity in $x[n]$, as shown in Fig. \ref{FIG:HidPSig}.
For the CSMP, the maximum hidden period $Q$ and the maximum number of iterations $L$ are set to $100$ and $20$, respectively.
For a fair comparison, the energies of these sinusoidal components are used as the references to represent the strengths of these hidden periods in $x[n]$, as shown in Fig. \ref{FIG:IdenHidPs}(a) and \ref{FIG:IdenHidPs2}(a).
The results of identifying the hidden periods with these methods are discussed in detail as follows.

\begin{figure}[t]
        \begin{minipage}[h]{1.0\linewidth}
          \centering
          \centerline{\includegraphics[width=9.5cm]{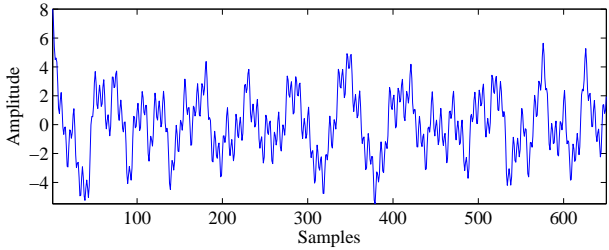}}
        \end{minipage}
        \caption{The synthetic signal of length $650$ generated by summing several sinusoidal components with the following periods: $5,12,25,26,57,58,70$, and $85$.}
        \label{FIG:HidPSig}
\end{figure}

\begin{figure}[htb]
        \begin{minipage}[h]{1.0\linewidth}
          \centering
          \centerline{\includegraphics[width=9.5cm]{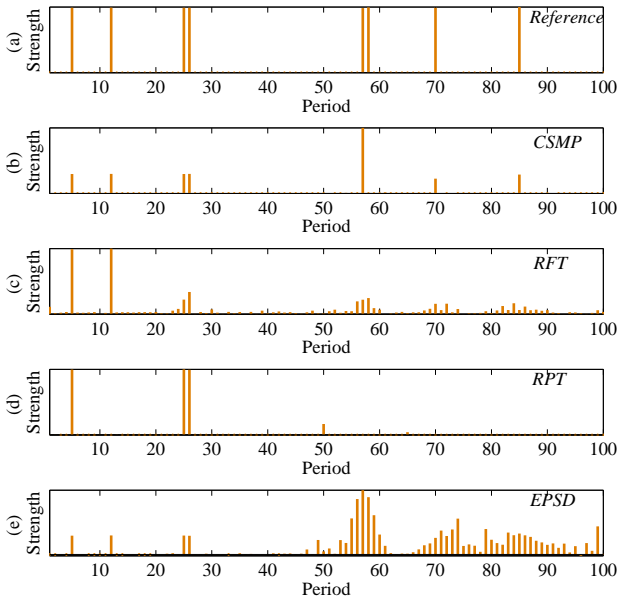}}
        \end{minipage}
        \caption{Hidden period identification for the synthetic signal of length $650$ using several approaches.}
        \label{FIG:IdenHidPs}
\end{figure}

Fig. \ref{FIG:IdenHidPs}(b) shows the results of identifying hidden periods of a signal of length $650$ with the CSMP algorithm.
The CSMP can correctly identify most of the hidden periods, but the hidden period $58$ is missing.
Because the signal length $650$ is too short, some overlap exists between the Ramanujan subspaces $\mathcal{S}_{57}$ and $\mathcal{S}_{58}$.
More specifically, for the case of the signal in (\ref{EQ:SYSIG}), there is an overlap between the CCSs $\mathcal{G}_{57,1}$ and $\mathcal{G}_{58,1}$,
where $\mathcal{G}_{57,1} \subset \mathcal{S}_{57}$ and $\mathcal{G}_{58,1} \subset \mathcal{S}_{58}$.
The CSMP therefore fails to correctly distinguish the adjacent hidden periodic components $x_{57} \in \mathcal{G}_{57,1}$ and $x_{58} \in \mathcal{G}_{58,1}$.
The periodic component $x_{58}$ is captured by the CCS $\mathcal{G}_{57,1}$, and hence, the hidden period $57$ in Fig. \ref{FIG:IdenHidPs}(b) achieves more strength than the other hidden periods.
However, the CSMP can correctly identify all these hidden periods including $57$ and $58$ of the signal in (\ref{EQ:SYSIG}) of length $1950$, as shown in Fig. \ref{FIG:IdenHidPs2} (b) because the overlap between the CCSs $\mathcal{G}_{57,1}$ and $\mathcal{G}_{58,1}$ decreases when the signal length increases; hence, the CSMP can achieve perfect identification results compared with the reference in Fig. \ref{FIG:IdenHidPs2}(a).
Although the hidden periods $25$ and $26$ are also adjacent hidden periods, the signal length $650$ is long enough for the CSMP to correctly distinguish them.
In general, a larger hidden period requires a longer signal length to eliminate the overlap between the corresponding CCS and other CCSs.

\begin{figure}[t]
        \begin{minipage}[h]{1.0\linewidth}
          \centering
          \centerline{\includegraphics[width=9.5cm]{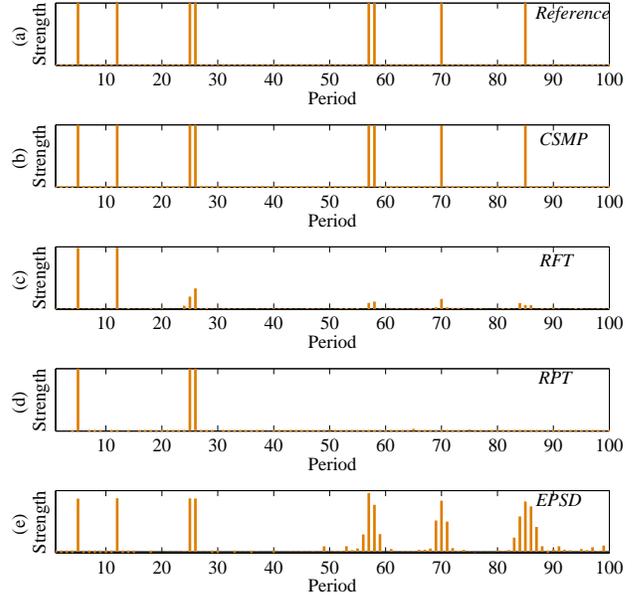}}
        \end{minipage}
        \caption{Hidden period identification for the synthetic signal of length $1950$ using several approaches.}
        \label{FIG:IdenHidPs2}
\end{figure}

Fig. \ref{FIG:IdenHidPs}(c) and Fig. \ref{FIG:IdenHidPs2}(c) show the results of identifying hidden periods with the RFT algorithm.
The RFT explores the periodic structure of the signal $x[n]$ by projecting it onto each Ramanujan sum $c_q[n]$, for $q=1,\cdots, 100$.
The square values of these projection coefficients are used to represent the strength of the hidden periods.
Unfortunately, the Ramanujan sums $c_q[n]$ cannot represent the whole Ramanujan subspace $\mathcal{S}_q$,
as $c_q[n]$ is only one of the $\phi(q)$ basis vectors of $\mathcal{S}_q$ (see \cite{Vaidyanathan2014_2} for more details).
In both Fig. \ref{FIG:IdenHidPs}(c) and Fig. \ref{FIG:IdenHidPs2}(c), the smaller hidden periods $\{5,12,25,26\}$ can be identified by the RFT to a certain extent, as the signal lengths $650$ and $1950$ are long enough to eliminate the overlaps among these `smaller' Ramanujan subspaces.
In fact, the Ramanujan subspaces $\mathcal{S}_{5},\mathcal{S}_{25},\mathcal{S}_{26}$ are manually orthogonal, as $5,25,26$ are divisors of the lengths $650$ and $1950$.
However, the RFT only captures part of the strengths of the larger hidden periods $57,58,70,85$ of the signals of lengths $650$ and $1950$, due to the overlaps among these larger Ramanujan subspaces.
In addition, many false periods are found around the true hidden periods in Fig. \ref{FIG:IdenHidPs}(c) because the Ramanujan subspaces associated with these periods are not mutually orthogonal for a finite signal.
When a signal $x[n]$ of longer length $1950$ is used to identify the hidden periods with the RPT, as shown in Fig. \ref{FIG:IdenHidPs2}(c), the strengths of the false periods decrease because the overlaps among these Ramanujan subspaces also decrease.
In fact, the RFT can be viewed as a special case of the EPSD algorithm for identifying the hidden periods, and the only difference between them is that the RFT can only capture the Ramanujan subspace $\mathcal{S}_q$ in a single dimension but the EPSD can capture the whole $\mathcal{S}_q$.

Fig. \ref{FIG:IdenHidPs}(d) and \ref{FIG:IdenHidPs2}(d) show the results of identifying hidden periods using the RPT algorithm.
In the RPT, the signal $x[n]$ is projected onto a series of orthogonal Ramanujan subspaces, and the energies of the projections of $x[n]$ into these subspaces are used to represent the strengths of the hidden periods.
However, the RPT only generates the subspace $\mathcal{S}_q$ where the period $q$ is just the divisor of the signal length.
This means that the RPT only identify the periods that are the divisors of the signal length.
Specifically, as the hidden periods $5,25$ and $26$ are divisors of the signal lengths $650$ and $1950$, respectively, they can be correctly identified by the RPT as shown in Fig. \ref{FIG:IdenHidPs}(d) and Fig. \ref{FIG:IdenHidPs2}(d).
The Ramanujan subspace $\mathcal{S}_{50}$, where the period $50$ is also a divisor of the signal length $650$, has some overlaps with other subspaces.
This results in some signal components being captured by $\mathcal{S}_{50}$, and hence, the false hidden period $50$ can be found in Fig. \ref{FIG:IdenHidPs}(d).
The overlap between $\mathcal{S}_{50}$ and other Ramanujan subspaces decreases when the signal length increases,
and hence, the false hidden period $50$ can be found in Fig. \ref{FIG:IdenHidPs2}(d).
Because only the Ramanujan subspaces whose periods are the divisors of the signal length are mutually orthogonal,
the EPSD can correctly identify the hidden periods $5,25,26$.
The hidden period $12$ is a smaller period compared with the signal signal length $650$ and $1950$ and can also be identified by the EPSD.
However, the other Ramanujan subspaces, $\mathcal{S}_{57},\mathcal{S}_{58},\mathcal{S}_{70}$, and $\mathcal{S}_{85}$, show some overlaps, and hence, many false hidden periods are found in Fig. \ref{FIG:IdenHidPs}(e) and \ref{FIG:IdenHidPs2}(e) around these periods.
When the signal length increases, the overlaps among these Ramanujan subspaces decrease and the false hidden periods also decrease, as shown in Fig. \ref{FIG:IdenHidPs2}(e).

In summary, due to the limited length of the signal, overlaps exist among the Ramanujan subspaces, and many methods fail to correctly identify hidden periods of the signal, such as RFT, RPT and EPSD.
Compared with these methods, the CSMP can identify most hidden periods of the signal and achieves better performance, which is attributed to its periodicity and energy strategies and the representation of the Ramanujan subspace with the CCSs.

\subsection{Signal approximation}

With the CSMP algorithm, a signal can be approximated by a sum of periodic components.
Assuming that $Q$ is the maximum period of the periodic components of the signal $x$,
a series of periodic components with period $q \in [1,Q]$ are selected to approximate the signal $x$.
These periodic components can adaptively capture the periodic structure of the signal.
The CSMP can quickly achieve convergence if the decomposed signal contains an obvious periodic structure, whereas the convergence speed of the CSMP is relatively slow if there is no obvious periodic structure in the signal.

\begin{figure}[t]
        \begin{minipage}[h]{1.0\linewidth}
          \centering
          \centerline{\includegraphics[width=9.5cm]{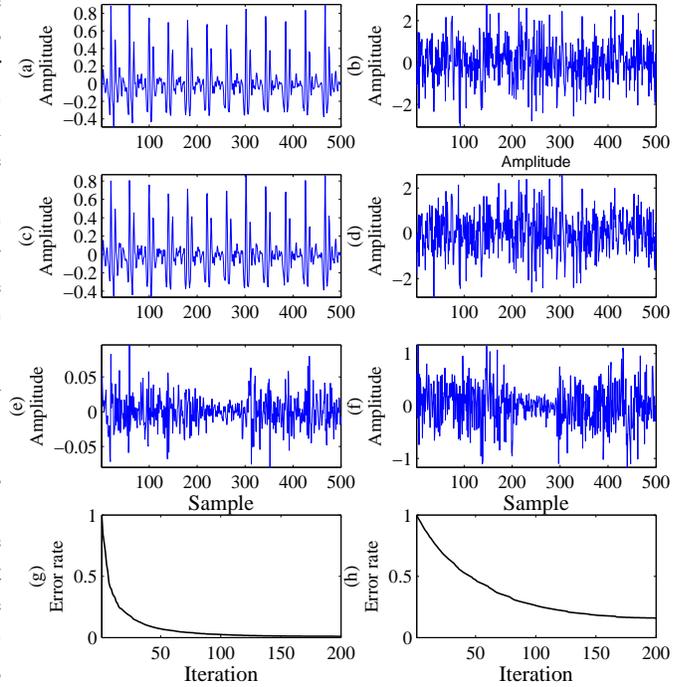}}
        \end{minipage}
        \caption{Approximation of the speech and white noise with CSMP. (a) Speech; (b) White noise; (c) Approximation of the speech in (a); (d) Approximation of the white noise in (b); (e) Residual from approximating the speech; (f) Residual from approximating the white noise; (g) Error rate for the speech approximation; (h) Error rate for the noise approximation.}
        \label{FIG:Dec_Speech_Noise}
\end{figure}

Fig. \ref{FIG:Dec_Speech_Noise} shows examples of approximating speech in Fig. \ref{FIG:Dec_Speech_Noise}(a) and white noise in Fig. \ref{FIG:Dec_Speech_Noise}(b) by using the CSMP algorithm.
In the CSMP, the maximum period $Q$ and the total number $L$ of iterations are set to $300$ and $200$, respectively.
The approximation and residue of the speech are shown in Fig. \ref{FIG:Dec_Speech_Noise}(c) and (e), respectively.
Compared with the approximation, the residual signal has a smaller amplitude, which means that most of the speech can be characterized by periodic components with period $q \in [1,300]$.
The final approximation error of the speech rate is $0.010$.
For the white noise, however, its residual signal shown in Fig. \ref{FIG:Dec_Speech_Noise}(f) is relatively large compared to its approximation, as shown in Fig. \ref{FIG:Dec_Speech_Noise}(d), and its approximate error rate is $0.266$.
Fig. \ref{FIG:Dec_Speech_Noise}(g) and (h) show that the error rates associated with approximating the speech and white noise decrease as the number of iterations increases.
The approximation of the speech can quickly achieve convergence after almost $100$ iterations,but the approximation of the white noise is relatively slow because there is an obvious periodic structure in the speech but not in the white noise.
Moreover, the approximation error rate cannot further decrease even if the number of iterations increases, as the maximum period $300$ is less than that of both signals.
We will present another example to explain this problem.

\begin{figure}[t]
        \begin{minipage}[h]{1.0\linewidth}
          \centering
          \centerline{\includegraphics[width=9.5cm]{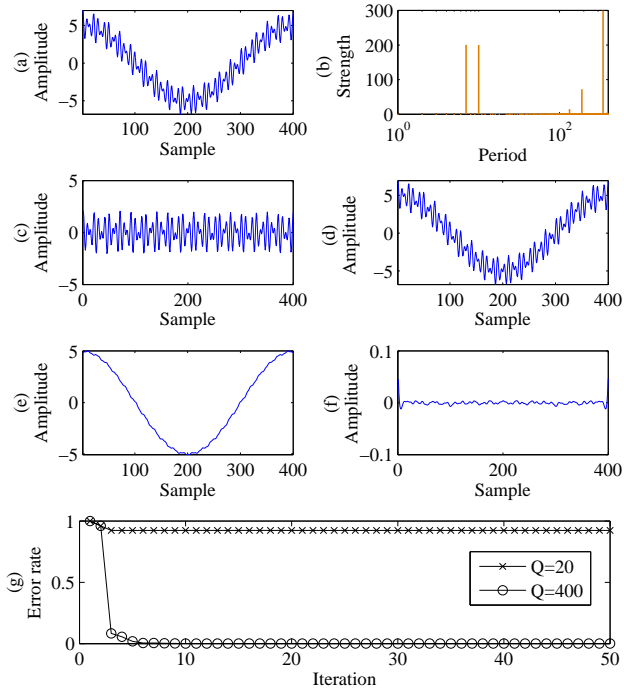}}
        \end{minipage}
        \caption{The synthetic signal of length $400$ with hidden periods $7,10$ and larger periods is approximated with different maximum periods $Q$.
(a) Synthetic signal; (b) Strength of the hidden periods; (c) Approximation with the maximum period $Q=20$; (d) Approximation with the maximum period $Q=400$; (e) Residual of the approximation in (c); (f) Residual of the approximation in (d); (g) The error rates of two approximations with $Q=20$ and $Q=400$.}
        \label{FIG:Maximum_Q}
\end{figure}

A synthetic signal of length $400$ is generated by superimposing several periodic signal with periods $7,10$ and other larger periods, as shown in Fig. \ref{FIG:Maximum_Q}(a).
The strengths of these hidden periods are shown in Fig. \ref{FIG:Maximum_Q}(b) in logarithmic coordinates.
The synthetic signal is approximated with different maximum periods $Q$, $20$ and $400$.
The final error rates for the different maximum periods are $0.923$ and $1.477\text{E-}6$.
The residual signals of the approximations with different maximum periods are shown in Fig. \ref{FIG:Maximum_Q}(e) and (f).
For the case of $Q=20$, the residue in Fig. \ref{FIG:Maximum_Q}(e) contains hidden periodic components whose periods far greater than $20$.
These periodic components cannot be approximated by the signals with periods $q \in [1,20]$, even if the number of iterations increases, as shown in Fig. \ref{FIG:Maximum_Q}(g).
However, the signal in Fig. \ref{FIG:Maximum_Q}(a) can be approximated well for the case $Q=200$, as shown in Fig. \ref{FIG:Maximum_Q}(d), and its residue in Fig. \ref{FIG:Maximum_Q}(f).
Therefore, the maximum period $Q$ must be greater than the maximum period $Q$ of the periodic components so that a good approximation is achieved.

\subsection{Tracing a time-varying period with shifted CSMP}

The periodicity of some signals varies with time, such as in an inverse chirp signal, speech and music.
The periods of these signals can change with time or be present for a short duration.
To detect the periodicity in these signals, the CSMP can be performed in a shifted rectangular window over the signal, which is referred to as shifted CSMP, similar to the short-time Fourier Transforms.
Before performing the shifted CSMP, a suitable window size $L$ needs to be chosen.
Assuming that the maximum period $Q$ of the signal in the shifted window is known, the window size $L$ must satisfy $L>Q$.
We present two examples to demonstrate the capability of our method for tracing a time-varying period.

\begin{figure}[t]
        \begin{minipage}[h]{1.0\linewidth}
          \centering
          \centerline{\includegraphics[width=9.5cm]{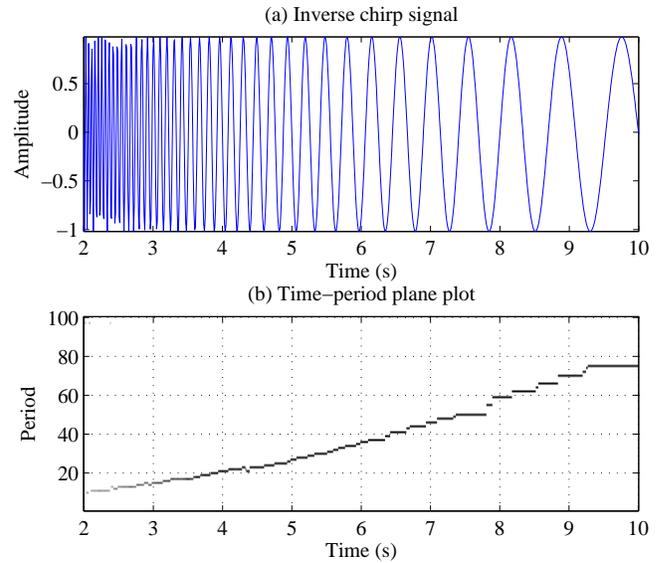}}
        \end{minipage}
        \caption{Period detection for an inverse chirp signal. (a) Inverse chirp signal; (b) Time-period plane using the shifted CSMP.}
        \label{FIG:IChirp}
\end{figure}

\begin{figure}[b]
        \begin{minipage}[h]{1.0\linewidth}
          \centering
          \centerline{\includegraphics[width=9.5cm]{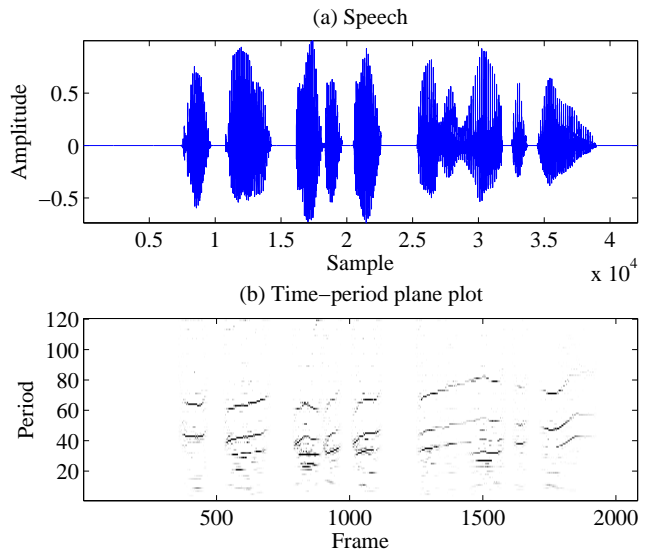}}
        \end{minipage}
        \caption{Period detection for speech. (a) Speech; (b) Time-period plane in the frame index using the shifted CSMP.}
        \label{FIG:Speech}
\end{figure}

In the first example, the inverse chirp signal \cite{Tenneti2015P, Vaidyanathan2015E} is defined by
\begin{align*}
	x(t)=sin(1/at),\text{ }  t \in [2,10]
\end{align*}
where $a=0.01/2 \pi$.
The discrete inverse chirp signal $x[n]$ is sampled from $x(t)$ every $0.01$ s, as shown in Fig. \ref{FIG:IChirp}(a).
In the shifted CSMP, the maximum period $Q$ and the window size $W$ are set to $100$ and $150$, respectively.
In Fig. \ref{FIG:IChirp}(b), the time-period plane shows the result of tracking the period of the inverse chirp signal with the shifted CSMP.
The figure clearly shows that the period of the signal varies with time from $10$ to $75$ (corresponding to the time period range from $0.10$ to $0.75$ s).
The example illustrates that the shifted CSMP can correctly trace a period that varies with time, even with a fixed-length window.

The speech has more complex periodic behaviors.
Multiple pitch detection for speech is the core of a broad range of applications \cite{Roux2007}.
In the second example, a clean speech of length $2.6$ s with a sample rate of $8000$ Hz is considered, as shown in Fig. \ref{FIG:Speech}(a).
In the shifted CSMP, the maximum period $Q$ and the window size $W$ are set to $120$ and $512$, respectively.
The number of iterations of the CSMP in each shifted window is set to $10$, which means that $10$ dominant periodic components are chosen in each window.
The time-period plane in Fig. \ref{FIG:Speech}(b) shows the result of tracking periods of the speech in the frame index.
There are obvious periodic structures (pitches) when the voice is active in the speech that vary with time.

\section{Conclusion}
\label{SEC:CON}

In this paper, we presented a new method called CSMP to exploit the periodicity of signals and their decomposition according to the periods.
Apart from identifying the hidden periods from signals, the CSMP is capable of detecting the changes of the period with time.
More generally, the CSMP can decompose any signal into a series of periodic components and residuals, which provides another view for exploiting the structure of the signal that is different from traditional frequency analysis.
Similar to the time-frequency representation, the shifted CSMP also presents the representation of signals in the time-period plane, which completely reveal the hidden periods of the signal that vary with time.
Different from the traditional method based on a greedy strategy, the CSMP can be effectively performed over CCSs in two stages without using the dictionary.
In the first stage, the dominant hidden period is estimated based on the periodic metric of the signal in each Ramanujan subspace, based on which the dominant CCS is chosen for updating the signal.
Finally, we applied the CSMP to several examples for hidden period identification, signal approximation and pitch detection in speech to illustrate the effectiveness of the proposed method.

\section* {Appendix}

\subsection{Orthogonal projection in the complex conjugate subspace}
\label{SEG:APP-A}

Let $\mathcal{G}$ be the conjugate subspace spanned by $\{g,\bar{g}\}$, i.e., $\mathcal{G}=span\{g,\bar{g}\}$, and the matrix $G=[g,\bar{g}]$.
Then, the complex projection coefficients of the real signal $x$ into $\mathcal{G}$, $\pmb{\alpha}=[\alpha_1, \alpha_2]^T$, can be calculated by
\begin{align*}
	\pmb{\alpha} &= (G^H G)^{-1} G^H x
\end{align*}
Because
\begin{align*}
	(G^H G)^{-1}=\frac{1}{1-|\langle g,\bar{g} \rangle|^2} \left[
            \begin{array}{cc}
                1,& - \langle g,g^* \rangle \\
                -\langle g,g^* \rangle,&   1
            \end{array}
            \right]
\end{align*}
and then, we have
\begin{align*}
    \pmb{\alpha} = \left[
        \begin{array}{c}
            \alpha_1 \\
            \alpha_2
        \end{array}
     \right]
        = \frac{1}{1-|\langle g,\bar{g} \rangle|^2} \left[
            \begin{array}{c}
                \langle g,x \rangle - \langle g,\bar{g} \rangle \langle \bar{g},x \rangle \\
                \langle \bar{g},x \rangle - \langle g,\bar{g} \rangle \langle g,x \rangle
            \end{array}
            \right]
\end{align*}
Note that as $x$ is the real signal, the two complex projection coefficients in the above equation are a complex pair, that is,
\begin{align*}
    \alpha_1 = \overline{\alpha}_2
\end{align*}
Thus, the projection $x_{\mathcal{G}}$ of $x$ into the conjugate subspace $\mathcal{G}$ is
\begin{align*}
	x_{\mathcal{G}} &= G(G^H G)^{-1} G^H x \\
		&= G \pmb{\alpha} \\
		&= \alpha_1 g + \alpha_2 \bar{g} \\
		&=  \alpha_1 g + \bar{\alpha}_1 \bar{g} \\
		&= 2 \text{Re}(\alpha g)
\end{align*}
where $\alpha=\alpha_1=\bar{\alpha}_2$ and
\begin{align*}
    \alpha = \frac{1}{1-|\langle g,\bar{g} \rangle|^2} \left(\langle g,x \rangle - \langle g,\bar{g} \rangle \langle \bar{g},x \rangle \right)
\end{align*}
Moreover, we have
\begin{align*}
	\|x_{\mathcal{G}}\|^2 = \|\pmb{\alpha}\|^2 = |\alpha_1|^2+|\alpha_2|^2=2|\alpha|^2
\end{align*}

\ifCLASSOPTIONcompsoc
  \section*{Acknowledgments}
\else
  \section*{Acknowledgment}
\fi
 This work was supported in part by the Major Research plan of the National Natural Science Foundation of China (No. 91120303), National Natural Science Foundation of China (No. 91220301), Natural Science Foundation of Heilongjiang Province of China (No. F2015012), and Academic Core Funding of Young Projects of Harbin Normal University of China (No. KGB201225).

\ifCLASSOPTIONcaptionsoff
  \newpage
\fi



%

%

\begin{IEEEbiography}{Deng Shi-Wen}
Shiwen Deng received a B.E degree from the Institute of Technology from Jia Mu Si University, JiaMuSi, China, in 1997, an M.E from The School of Computer Science from Harbin Normal University, Harbin, China, in 2005, and a Ph.D from the school of Computer Science from the Harbin Institute of Technology in 2012. Currently, he is with the School of Mathematical Sciences, Harbin Normal University, Harbin, China. His research interests lie in the areas of speech and audio signal processing, including content-based audio analysis, noise suppression, and speech/audio classification/detection.
\end{IEEEbiography}

\begin{IEEEbiography}{Han Ji-Qing}
Jiqing Han received B.S. and M.S. degrees in electrical engineering and a Ph.D. in computer science from the Harbin Institute of Technology, Harbin, China, in 1987, 1990, and 1998, respectively. Currently, he is the associate dean of the school of Computer Science and Technology, Harbin Institute of Technology. He is a member of IEEE, a member of the editorial board of the Journal of Chinese Information Processing, and a member of the editorial board of the Journal of Data Acquisition and Processing. Prof. Han is undertaking several projects with the National Natural Science Foundation, 863Hi-tech Program, National Basic Research Program. He has won three Second Prize and two Third Prize awards in Science and Technology from the Ministry/Province. He has published more than 100 papers and 2 books. His research fields of expertise include speech signal processing and audio information processing.
\end{IEEEbiography}





\end{document}